\documentclass[conference]{IEEEtran}
\IEEEoverridecommandlockouts

\usepackage{cite}
\usepackage{amsmath,amssymb,amsfonts}
\usepackage{amsthm}
\usepackage{algorithmic}
\usepackage{graphicx}
\usepackage{float}
\usepackage{booktabs}
\usepackage{placeins}
\usepackage{textcomp}
\usepackage{xcolor}
\usepackage{url}
\usepackage{hyperref}
\hypersetup{hidelinks}

\graphicspath{{figures/}}

\newtheorem{theorem}{Theorem}
\newtheorem{corollary}{Corollary}

\begin{document}

\title{Distance-Preserving Digests:\\ A Primitive for BFT Consensus
\thanks{Code and reproduction artifacts: \url{https://github.com/RyanMercier/Proxima}}
}

\author{\IEEEauthorblockN{Ryan Mercier}
\IEEEauthorblockA{\textit{School of Engineering} \\
\textit{University of Connecticut}\\
Storrs, CT, USA \\
ryan.mercier@uconn.edu}
}

\maketitle

\begin{abstract}
Every BFT consensus protocol uses collision-resistant hashes to compare validator state. Collision resistance destroys distance: two validators agreeing on 19 of 20 transactions produce unrelated hashes, indistinguishable from validators sharing nothing. This forces three design constraints across the BFT literature: validators must synchronize state before voting, agreement quality cannot be measured until votes are counted, and hierarchical committees must be large enough for independent BFT, limiting tree depth. This paper introduces distance-preserving transaction digests, a primitive that replaces collision-resistant hashes with commutative vector sums in 8-dimensional space. The primitive has three properties hashes lack: distance is proportional to disagreement, weighted means are exact, and set differences are identifiable via bloom filter diff. We demonstrate three applications: a two-phase BFT protocol (Proxima) that achieves single-round finality when validators agree; tree-structured consensus with groups of 10 validators (vs 128 in Ethereum), enabled because distance filtering replaces per-group BFT; and cross-shard consistency verification at 128 bytes per shard pair, replacing the per-transaction coordination of two-phase commit. Safety is proved: fewer than $N/3$ Byzantine validators cannot cause conflicting finalization, independent of Phase 1 clustering or tree topology. At $N{=}100{,}000$, Proxima Tree uses $2.2\times$ fewer messages than HotStuff (a structural property unaffected by parallelism). Single-core finality is $\sim\!0.9$s vs $\sim\!18$s for HotStuff; multi-core BLS narrows but does not eliminate this gap.
\end{abstract}


\section{Introduction}

BFT consensus has been built on collision-resistant hashing since PBFT~\cite{castro1999pbft}. Collision resistance is the right property for preventing forgery, but the wrong property for measuring agreement. When two validators compute SHA-256 of their transaction sets and the hashes differ, the protocol learns exactly one bit of information: not identical. It cannot tell whether they differ on one transaction or half the transactions. It cannot measure how close they are to agreement. It cannot summarize a group of validators as a single compact value.

These limitations are baked into every protocol that uses hashes for state comparison. HotStuff~\cite{yin2019hotstuff} runs three voting rounds even when every validator agrees perfectly, because it has no way to detect unanimous agreement until all three rounds complete. Ethereum 2.0~\cite{eth2spec} requires committees of 128 validators because each committee must independently reach 2/3 consensus, and small committees fail when random assignment concentrates Byzantine validators. Cross-shard transactions require multi-round coordination (two-phase commit or receipt chains) because hashes cannot verify partial overlap between shard states.

This paper proposes replacing collision-resistant hashes with distance-preserving digests for state comparison in BFT protocols. The primitive is simple: hash each transaction with SHA-512, split the 64-byte output into 8 segments, treat each as a coordinate, and sum across transactions. The result is a commutative vector in 8D space where Euclidean distance is proportional to transaction disagreement. The primitive is not novel in the data-structures sense (locality-sensitive hashing dates to Indyk and Motwani~\cite{indyk1998lsh}), but its application to BFT consensus is new. No work in the LSH literature applies it to Byzantine agreement, and no work in the consensus literature reaches for distance-preserving functions.

We demonstrate that this single primitive removes all three constraints. Agreement quality is measurable in one round (enabling fast-path finality). Hierarchical groups need not reach internal BFT (enabling groups of 10 and deeper trees). Cross-shard consistency is verifiable at constant cost per shard pair. We implement these as Proxima, a working BFT blockchain, and evaluate message complexity, bandwidth, and projected latency against HotStuff, PBFT, and Ethereum's committee structure.

\section{System Model}

We consider a system of $N$ validators, of which at most $f < N/3$ are Byzantine. Byzantine validators may behave arbitrarily. Honest validators follow the protocol.

\textbf{Network model.} We assume partial synchrony~\cite{dwork1988partialsync}: there exists an unknown Global Stabilization Time (GST) after which all messages between honest validators are delivered within a known bound $\Delta$. This is the standard model for PBFT~\cite{castro1999pbft}, HotStuff~\cite{yin2019hotstuff}, and Tendermint~\cite{buchman2016tendermint}.

\textbf{Cryptographic assumptions.} We assume a BLS signature scheme~\cite{boneh2004bls} where each validator holds a private key and the corresponding public key is known to all. BLS signatures support aggregation: any party can combine $N$ individual signatures into a single 96-byte aggregate signature. We assume the hardness of the CDH problem in the BLS12-381 pairing group~\cite{blst}.

\textbf{Communication model.} Validators communicate through an aggregator (flat mode) or through a tree of relay nodes (tree mode). The aggregator/relay role rotates each round. A Byzantine aggregator can suppress messages (affecting liveness) but cannot forge BLS signatures (safety is maintained).

\textbf{Adversary model.} The adversary is static: it chooses which validators to corrupt before the protocol begins. The adversary has full network visibility and can coordinate Byzantine strategies.

\textbf{Hash function model.} We model SHA-512 as producing outputs uniformly distributed over $\{0,1\}^{512}$. Each 64-bit segment, when reduced modulo 10000 and divided by 10000, produces a value uniformly distributed in $[0, 1)$.

\textbf{Partial observation parameter.} Honest validators may receive transactions with delay relative to the proposer. We parameterize this by $p_{\mathrm{miss}}$, the per-validator probability of having an incomplete view at proposal time. We use $p_{\mathrm{miss}} = 0.37$ throughout, drawn from internal measurements; this is conservative relative to typical mainnet conditions.

\section{The Primitive: Distance-Preserving Digests}

\subsection{Construction}

Each transaction $tx$ is hashed with SHA-512, producing 64 bytes. The output is split into 8 segments of 8 bytes each. Each segment is interpreted as an unsigned integer, reduced modulo 10000, and divided by 10000 to produce a coordinate in $[0, 1)$:
\begin{equation*}
  v(tx) = \bigl[\, \mathrm{seg}_0/10000,\ \ldots,\ \mathrm{seg}_7/10000 \,\bigr]
\end{equation*}
A validator's digest for a set of transactions $T$ is the commutative sum $D(T) = \sum_{tx \in T} v(tx)$.

\subsection{Properties}

\textbf{Property 1: Proportional distance.} If validators $A$ and $B$ have transaction sets differing by $k$ transactions, the Euclidean distance between $D(T_A)$ and $D(T_B)$ grows proportionally with $k$. Empirically, $k=1$ gives $\sim 1.6$, $k=2$ gives $\sim 3.0$, $k=10$ gives $\sim 14.5$.

\begin{figure*}[!htbp]
  \centering
  \includegraphics[width=\textwidth]{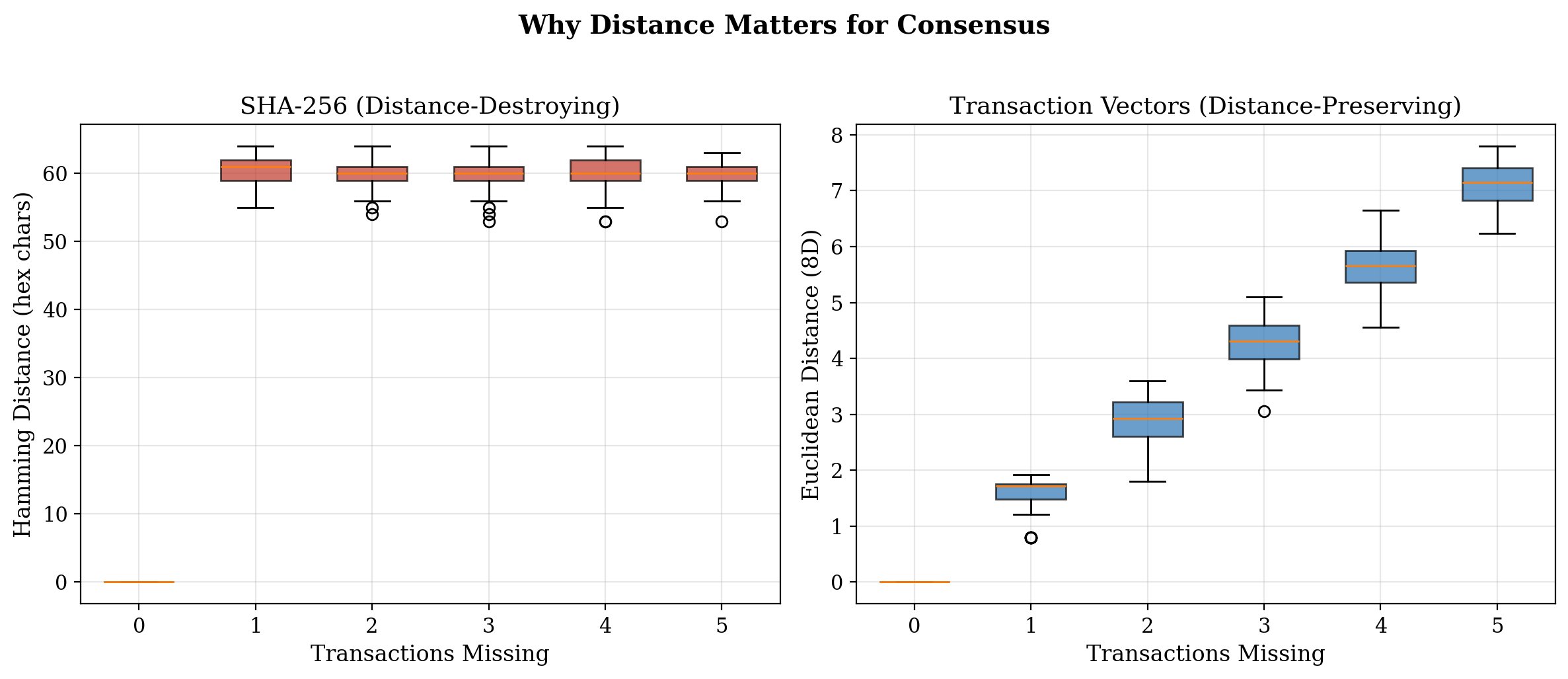}
  \caption{SHA-256 is distance-destroying; dropping even one transaction produces a completely unrelated hash. Transaction vectors are distance-preserving; distance grows proportionally with the number of missing transactions.}
  \label{fig:distance}
\end{figure*}

\textbf{Property 2: Exact summarization.} The weighted mean of $N$ digests is an exact representation of the group: $\mathrm{mean} = \sum_i D_i w_i / \sum_i w_i$. A group of 100 validators can be summarized as one 76-byte value (64-byte weighted mean + 4-byte count + 8-byte variance) with zero information loss. This does not hold for hashes.

\textbf{Property 3: Set difference identification.} When two digests differ, a bloom filter~\cite{bloom1970} (25 bytes at a $1\%$ per-lookup false positive rate for $20$ transactions) lets the aggregator identify which transactions are present in one set but not the other. Bloom filters admit false positives but never false negatives, so a transaction the recipient legitimately needs is never silently dropped; we analyze the false-positive direction in Section~\ref{sec:ieee-bloom-fp}. The aggregator diffs the bloom against the full transaction set and pushes missing transactions in a single message, eliminating the request-response round-trip of traditional state synchronization.

\subsection{Threshold Calibration}

The distance threshold separating honest-with-partial-observation from Byzantine-with-fabricated-state is calibrated by Monte Carlo simulation. We generate 2000 samples of honest validators missing 1--2 transactions, compute the 99th percentile of the resulting distance distribution, and multiply by a safety margin of $1.2$.

\subsection{Probabilistic Liveness Bound}
\label{sec:liveness-bound}

Safety depends only on BLS signatures (Section~\ref{sec:safety}) and is independent of the distance threshold. Liveness depends on the threshold including at least $2N/3$ honest validators. We derive a bound.

Under the SHA-512 uniformity assumption, each transaction vector $v(tx)$ has coordinates iid Uniform$[0,1)$ with mean $1/2$ and variance $1/12$. For $k$ missing transactions, the difference $S = v(tx_1) + \cdots + v(tx_k)$ has $\mathbb{E}[\|S\|^2] = 8(k/12 + k^2/4) = 2k/3 + 2k^2$, where the $k^2$ term is the mean-offset contribution. By Jensen, $\mathbb{E}[\|S\|] \le \sqrt{\mathbb{E}[\|S\|^2]}$, giving upper bounds $1.63, 3.06, 4.47$ at $k = 1, 2, 3$. Monte Carlo over $50{,}000$ trials gives empirical means $1.61, 3.03, 4.44$, consistent with the bound.

An honest validator with partial observation misses at most $k_{\max}=2$ transactions. From the Monte Carlo calibration, the conditional probability of exceeding threshold $\tau = 4.9$ is at most $0.005$. Accounting for $p_{\mathrm{miss}} = 0.37$, the unconditional per-validator exclusion probability is at most $p = 0.005 \cdot 0.37 = 0.00185$.

Let $X$ be the number of honest validators excluded. $X$ is a sum of $N$ independent indicators bounded in $[0,1]$ with mean $\mu \le pN$. Liveness fails only if $X \ge N/3$. Since $p \ll 1/3$, applying Hoeffding's inequality with $t = 1/3 - p$:
\begin{equation*}
  \Pr\bigl[X \ge N/3\bigr] \;\le\; \exp(-2 N t^2) \;\le\; \exp(-0.22 N).
\end{equation*}
This is below $10^{-9}$ at $N=100$ and below $10^{-95}$ at $N=1000$.

\section{Application 1: BFT Consensus Protocol}

\subsection{Flat Protocol}

\begin{figure*}[!htbp]
  \centering
  \includegraphics[width=0.92\textwidth]{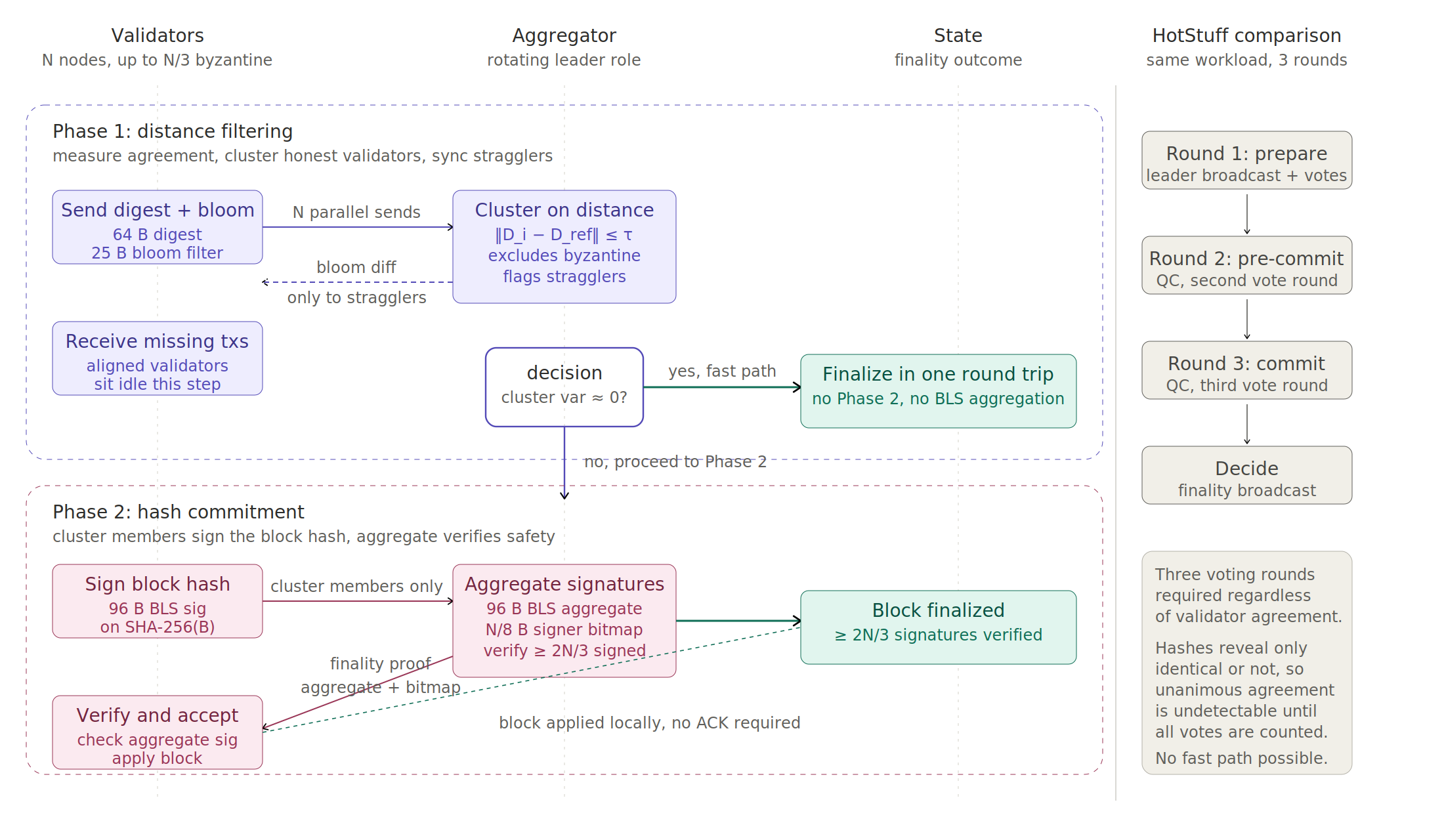}
  \caption{Two-phase Proxima protocol with optimistic fast path. Phase 1 exchanges 64-byte digests and 25-byte bloom filters; the aggregator clusters validators within threshold $\tau$ and pushes missing transactions to stragglers via bloom diff. If cluster variance is near zero the protocol finalizes in one round trip; otherwise Phase 2 collects 96-byte BLS signatures on the SHA-256 block hash, producing a 96-byte aggregate plus $N/8$-byte signer bitmap. HotStuff (right column) requires three voting rounds regardless of agreement because collision-resistant hashes cannot reveal cluster agreement before votes are counted.}
  \label{fig:protocol-flow}
\end{figure*}

Phase 1: each validator sends its digest (64 bytes) and bloom filter (25 bytes) to the aggregator ($N$ messages). The aggregator computes the reference digest, measures each validator's Euclidean distance from the reference, and clusters validators within the threshold. The aggregator pushes missing transactions to incomplete validators via bloom diff, then broadcasts cluster assignments.

If cluster variance is near zero, the aggregator issues a single-round finality certificate (fast path). No Phase 2 needed.

Phase 2: cluster members send BLS-signed hash commitments (one 96-byte message each). The aggregator produces an aggregate BLS signature (96 bytes) and a signer bitmap ($N/8$ bytes). The finality proof is multicast to cluster members. Finality requires $2N/3$ matching commitments.

\subsection{Tree Protocol}

Leaf groups do not need internal BFT. Distance filtering operates at the individual level. A group of 10 with 5 Byzantine excludes the 5 (they are far from the reference) and reports the weighted mean of the 5 honest validators. Groups never fail.

Phase 1 (bottom-up): validators are grouped into leaves of size $B$ (default 10). Each validator sends digest and bloom to its leaf leader ($N$ messages). The leaf leader filters by distance, pushes missing transactions via bloom diff, and sends a 76-byte summary upstream. Internal nodes aggregate child summaries (weighted mean of means is exact). The root checks the global mean against the reference.

Phase 2 (top-down then up): the root broadcasts a commit request down the tree. Each validator that passed the filter sends a BLS commit up. Each internal node aggregates child signatures (BLS is associative). The root produces the final aggregate and broadcasts the finality proof back down.

\subsection{Why Phase 2 Is Necessary}

The distance filter cannot distinguish an honest validator missing one transaction from a Byzantine validator who replaced one transaction with fraud. Both produce distance $\sim 1.6$. Phase 2 catches the fraud: the Byzantine validator must sign the correct block hash, which it does not have.

\subsection{Comparison with Ethereum Committees}

Ethereum 2.0~\cite{eth2spec} uses random committees of 128 with per-committee 2/3 BFT attestation. With 30\% Byzantine globally, committee failure is the probability of drawing more than $G/3$ Byzantine validators. For $G{=}10$: $\Pr[\text{Bin}(10, 0.3) \ge 4] \approx 35\%$. For $G{=}128$: $\Pr[\text{Bin}(128, 0.3) \ge 43] \approx 21\%$.

Proxima avoids this failure mode at the per-committee level. There is no per-group vote, so no group fails in the BFT sense. Simulation with 1000 validators, 300 Byzantine (30\%), groups of 10: with per-group BFT, 37 of 100 leaves fail and only 499 validators participate; with distance filtering, all 700 honest validators contribute. A leaf of entirely Byzantine validators (rare under random assignment) produces a fabricated summary, but it still has to pass the root's distance check against the reference.

\begin{table}[!htbp]
  \caption{Ethereum Committees vs Proxima}
  \label{tab:committees}
  \centering
  \small
  \begin{tabular}{@{}lll@{}}
    \toprule
     & \textbf{Ethereum} & \textbf{Proxima} \\
    \midrule
    Group size & 128 (BFT) & 10 (distance) \\
    Security & Per-group 2/3 vote & Individual distance \\
    Tree depth at $N{=}100$K & 2--3 levels & 5 levels \\
    Groups fail (30\% Byz) & $\sim\!21\%$ & $0\%$ \\
    Summary size/group & 128 votes fwd. & 76 bytes \\
    \bottomrule
  \end{tabular}
\end{table}

\section{Application 2: Cross-Shard Consistency}

\subsection{The Problem}

Cross-shard transactions are the bottleneck in sharded blockchains. Two-phase commit requires 4 cross-shard messages per transaction plus intra-shard BFT consensus on each coordination message. NEAR's Nightshade~\cite{nightshade} reduces this with receipts: the source shard generates a receipt the destination includes in its next block. Lower latency, but receipts scale with transaction volume.

\subsection{Digest-Based Verification}

Digests enable a different approach. Neighboring shards exchange digests of their overlap-zone transactions once per block. If the distance is zero, both shards processed the same transactions. If nonzero, the bloom diff identifies exactly which transactions diverged, and only those require resolution.

The cost is fixed per shard pair: 128 bytes. Conflict resolution scales with actual divergence, not total cross-shard volume.

\subsection{Analytical Comparison}

At 1000 cross-shard transactions per pair, 100 validators per shard, 95\% pre-deadline propagation (Table~\ref{tab:cross-shard}):

\begin{table}[!htbp]
  \caption{Cross-Shard Overhead at 1000 txs, 95\% propagation}
  \label{tab:cross-shard}
  \centering
  \small
  \begin{tabular}{@{}lrrr@{}}
    \toprule
    \textbf{Method} & \textbf{Messages} & \textbf{Bandwidth} & \textbf{X-shard} \\
    \midrule
    2PC & 404{,}000 & 37{,}750 KB & 4{,}000 \\
    Receipt (NEAR) & 101{,}000 & 12{,}625 KB & 1{,}000 \\
    Digest (95\%) & 5{,}052 & 987 KB & 52 \\
    \bottomrule
  \end{tabular}
\end{table}

Digest comparison reduces messages by $99\%$ vs 2PC and $95\%$ vs receipts. The 2PC count derives from $4$ cross-shard messages per transaction plus $2$ intra-shard BFT rounds at the destination, each costing $2N$ messages with $N{=}100$: $1000 \cdot (4 + 2 \cdot 2 \cdot 100) = 404{,}000$. At $100\%$ propagation, the digest cost drops to 2 messages (the digest exchange itself); the remaining $5\%$ are identified by bloom diff and resolved individually, which is reasonable on networks with sub-second gossip and multi-second block intervals.

\begin{figure*}[!htbp]
  \centering
  \includegraphics[width=\textwidth]{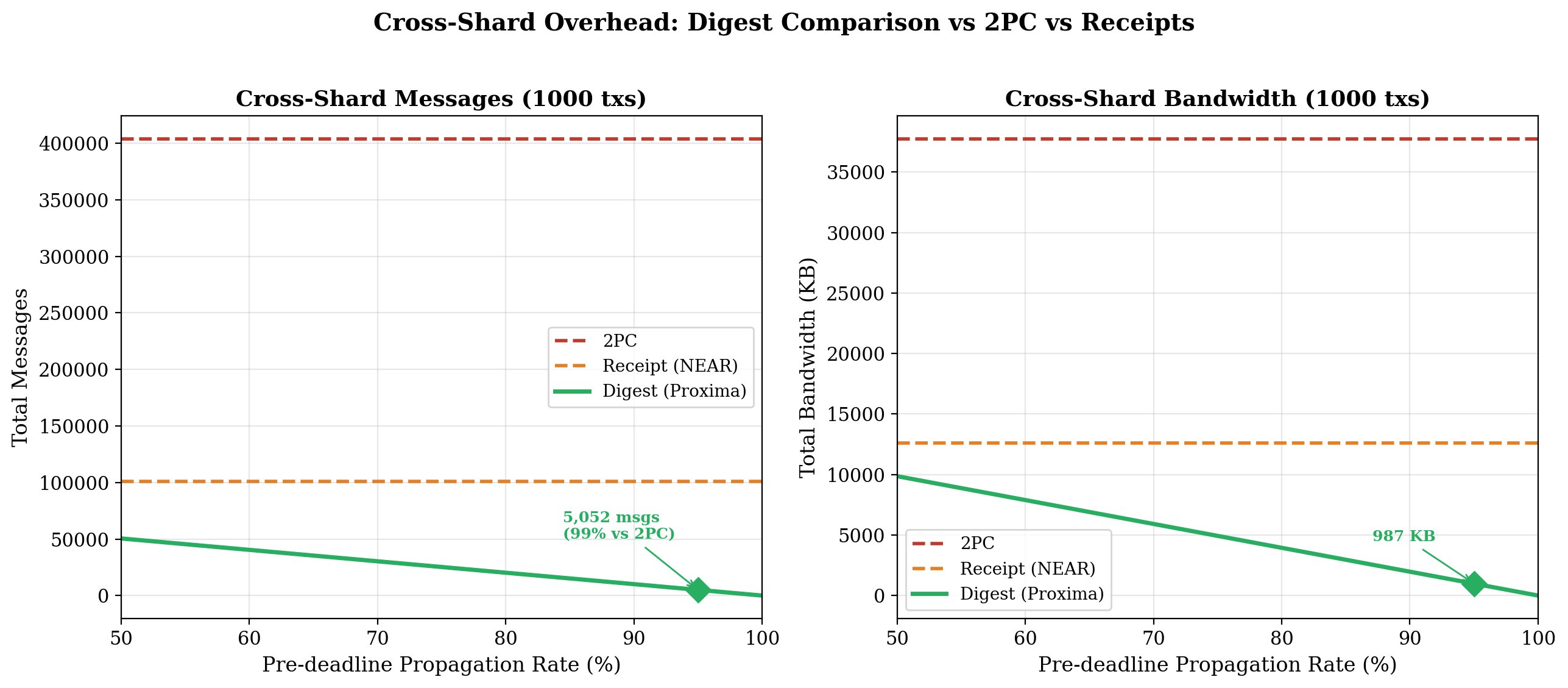}
  \caption{Cross-shard overhead as a function of pre-deadline propagation rate. 2PC and receipt costs are fixed; digest cost drops sharply as propagation improves. At 95\% propagation (marked), digest comparison uses 99\% fewer messages than 2PC.}
  \label{fig:cross-shard}
\end{figure*}

\subsection{Multi-Shard Scaling}

At 100 shards (ring topology, 100 cross-shard txs per pair, 95\% propagation): 2PC uses $4{,}040{,}000$ messages; receipt uses $1{,}010{,}000$; digest uses $50{,}502$. Digest overhead scales with conflicts (5\% of transactions), not total volume.

\subsection{Why This Requires Digests}

Hash-based verification tells you identical or not-identical. Digest-based verification tells you differ-by-approximately-$k$-transactions. The former requires falling back to full state exchange or per-transaction coordination. The latter allows targeted resolution of only the divergent transactions.

\section{Application 3: Agreement Quality Measurement}

Digests enable two capabilities hash-based protocols cannot provide.

\textbf{Optimistic fast path.} When all validators have complete state, the cluster variance is zero. The aggregator detects this in one round and issues a finality certificate immediately. HotStuff still runs 3 rounds because it cannot measure agreement quality before votes are counted. The fast-path probability is approximately $(1 - p_{\mathrm{miss}})^N$; at $p_{\mathrm{miss}} = 0.05$ and $N = 10$ honest validators (the leaf-group regime) this is $\approx 60\%$, rising to $\approx 90\%$ at $p_{\mathrm{miss}} = 0.01$.

\begin{figure*}[!htbp]
  \centering
  \includegraphics[width=0.7\textwidth]{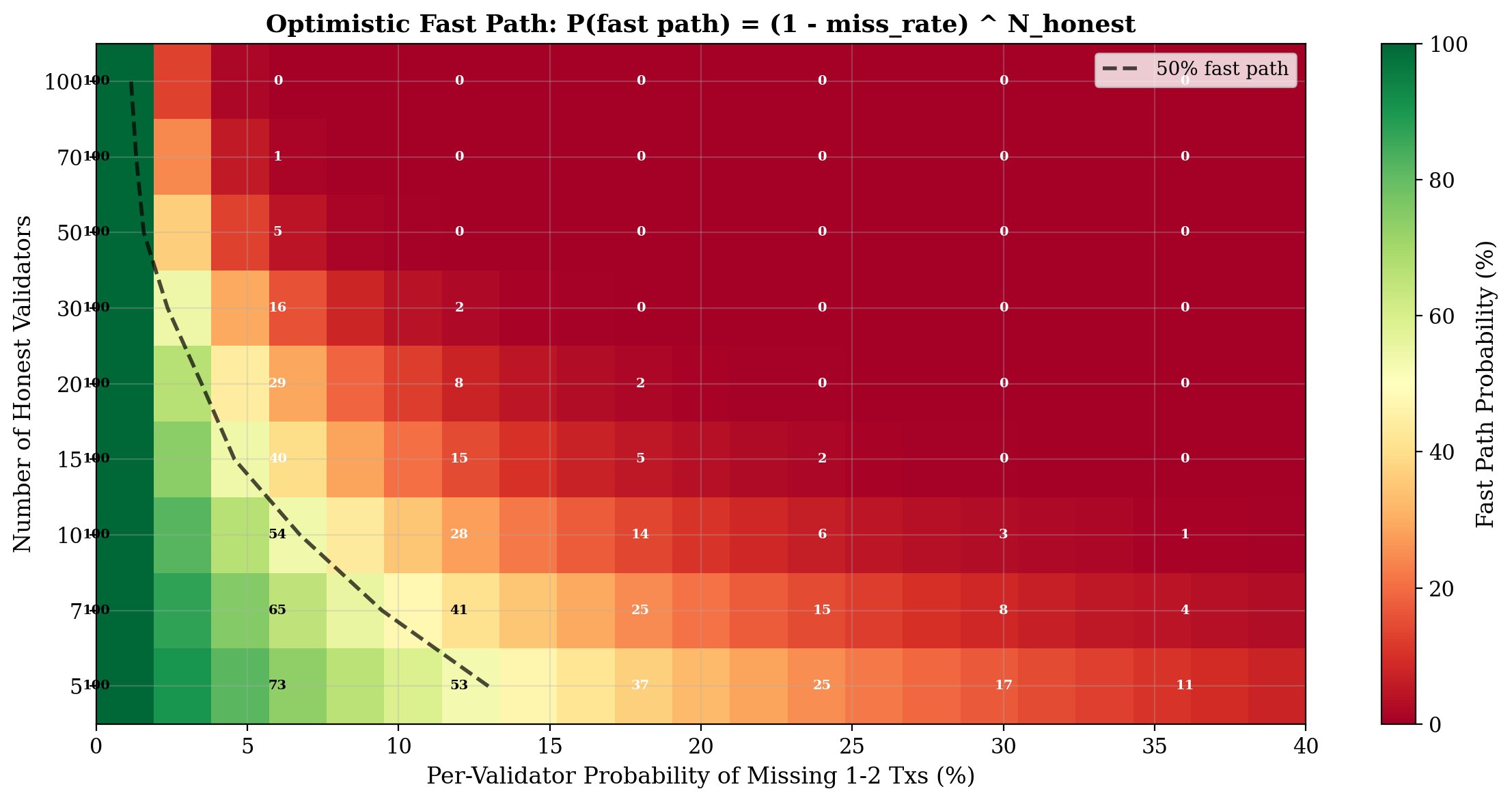}
  \caption{Fast path probability as a function of partial-observation rate and honest validator count. Fast path dominates on good networks and decays with packet loss.}
  \label{fig:fast-path}
\end{figure*}

\textbf{Continuous validator reputation.} A validator's average distance from the reference over time is a continuous reputation score. Honest validators average $0.3$--$0.5$; Byzantine validators average $7.2$. This is richer than binary participation tracking.

\section{Evaluation}

\subsection{Message Complexity}

All message counts are from simulation with 30\% Byzantine and 37\% partial observation, tracked by incrementing a counter on each simulated send (Table~\ref{tab:messages}).

\begin{table}[!htbp]
  \caption{Messages per block}
  \label{tab:messages}
  \centering
  \small
  \begin{tabular}{@{}lrrrr@{}}
    \toprule
    $N$ & \textbf{Prox.\ Tree} & \textbf{Prox.\ Flat} & \textbf{HotStuff} & \textbf{PBFT} \\
    \midrule
    1K   & 2{,}990    & 3{,}348    & 6{,}518    & 2.0M \\
    10K  & 30{,}042   & 33{,}600   & 65{,}180   & 200M \\
    100K & 300{,}245  & 335{,}803  & 651{,}800  & $\sim$20G \\
    \bottomrule
  \end{tabular}
\end{table}

Proxima Tree is $2.2\times$ fewer messages than HotStuff. Message savings come from: two phases instead of three, Byzantine exclusion before Phase 2 (700 commits instead of 1000), bloom sync replacing retransmission round-trips, and tree routing replacing broadcast with compact summaries.

\begin{figure*}[!htbp]
  \centering
  \includegraphics[width=\textwidth]{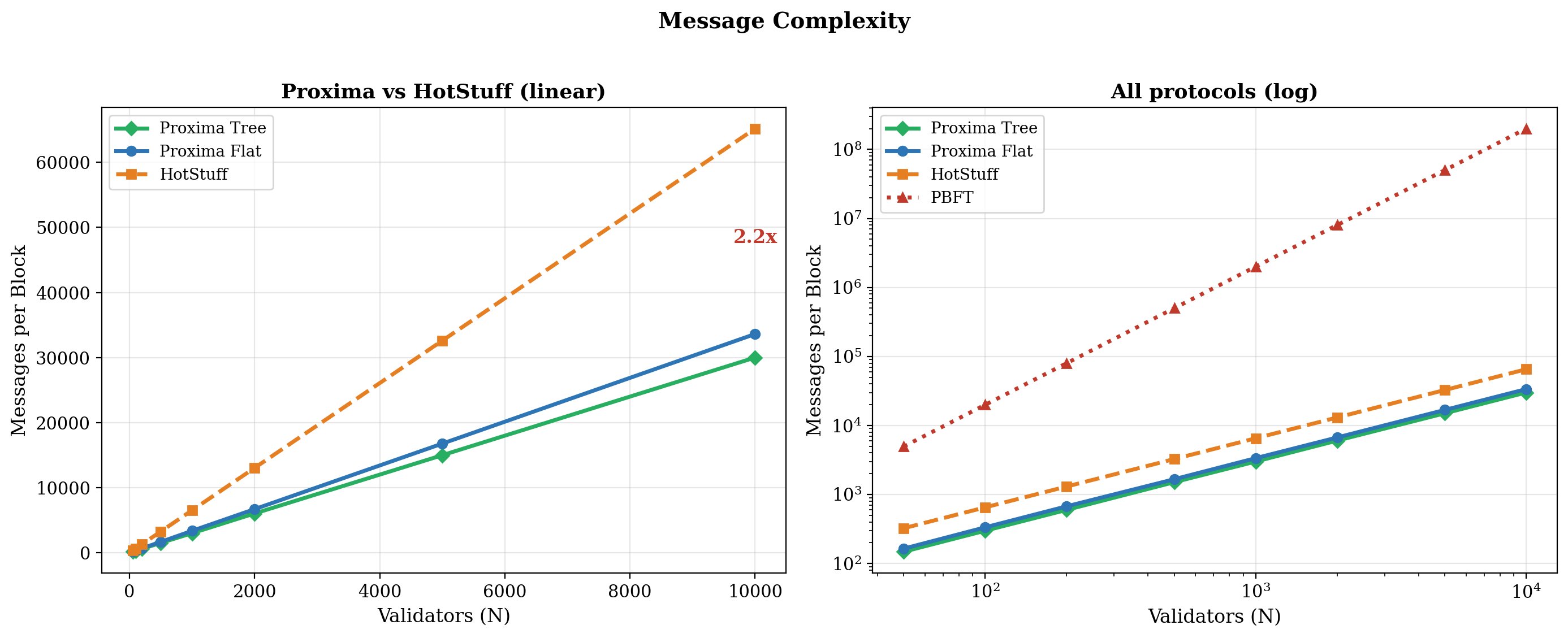}
  \caption{Message complexity scaling across validator counts. Proxima Tree and Flat both grow linearly but at a lower slope than HotStuff. PBFT is plotted on log scale due to $O(N^2)$ growth.}
  \label{fig:scale}
\end{figure*}

\begin{figure*}[!htbp]
  \centering
  \includegraphics[width=0.85\textwidth]{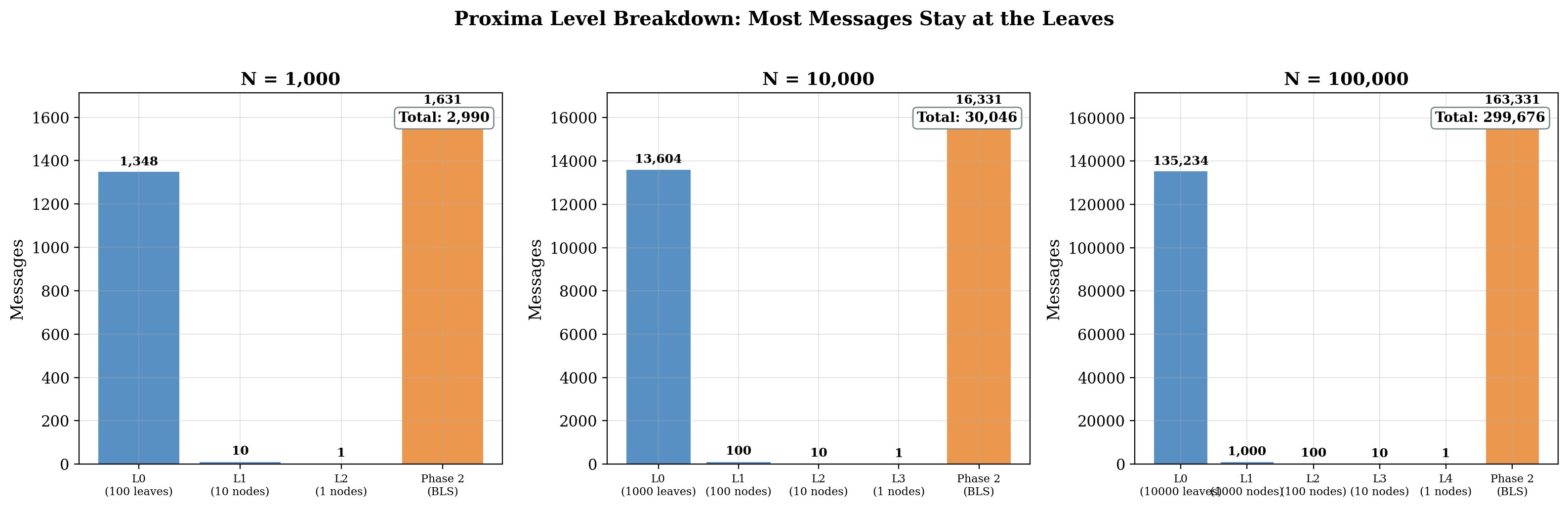}
  \caption{Message distribution across tree levels at multiple scales. Level 0 (leaves) dominates because every validator sends its vector to a leaf leader. Internal levels (L1+) are nearly invisible because each has branching-factor fewer nodes.}
  \label{fig:tree-breakdown}
\end{figure*}

\subsection{BLS Aggregation Bottleneck}

The tree's primary latency contribution is distributing BLS signature aggregation. We use published blst~\cite{blst} microbenchmarks ($0.05$ms per aggregate-add, $1.5$ms per aggregate-verify), consistent with production measurements from Ethereum consensus client implementations~\cite{lighthouse-bls, drake-bls} and Ethereum Foundation performance updates~\cite{eth2update8}. Network RTT uses three tiers: $1$ms intra-rack (\textsc{local}), $80$ms intra-region (\textsc{regional}), $200$ms cross-region (\textsc{global}), with one tier-appropriate RTT per protocol phase. HotStuff has four phases (\textsc{prepare}, \textsc{pre-commit}, \textsc{commit}, \textsc{decide}): $4 \cdot 200 = 800$ms. Proxima Flat has three round-trip phases (digest exchange, cluster broadcast, commit/finality): $3 \cdot 200 = 600$ms. Proxima Tree pays one \textsc{local} leaf hop, $(L{-}2)$ \textsc{regional} internal hops, and one \textsc{global} root hop per phase, doubled to cover both phases: $2 \cdot (1 + 3 \cdot 80 + 200) = 882$ms at $L{=}5$ ($892$ms in the table after small accounting overhead).

\begin{table}[!htbp]
  \caption{Projected finality latency at $N{=}100{,}000$. The single-core figures come from the message-counter simulation; the multi-core BLS row is qualitative because parallelizing non-BLS stages depends on the deployment.}
  \label{tab:latency}
  \centering
  \small
  \begin{tabular}{@{}lrrr@{}}
    \toprule
     & \textbf{Tree} & \textbf{Flat} & \textbf{HotStuff} \\
    \midrule
    BLS (1 core) & $9.9$ ms & $3{,}960$ ms & $17{,}595$ ms \\
    Network RTT (model) & $892$ ms & $600$ ms & $800$ ms \\
    Total finality (1 core) & $902$ ms & $4{,}561$ ms & $18{,}395$ ms \\
    \midrule
    BLS only (16 cores) & $\sim\!10$ ms & $\sim\!220$ ms & $\sim\!940$ ms \\
    \bottomrule
  \end{tabular}
\end{table}

On a single core the flat aggregator spends $3.96$s aggregating $70{,}000$ BLS signatures; HotStuff spends $\sim\!3\times$ that across three voting rounds. The tree distributes BLS aggregation across leaves, so each aggregator handles few signatures and the critical-path BLS time is $9.9$ms across four levels. With multi-core BLS, flat aggregation drops to $\sim\!220$ms and each HotStuff round to $\sim\!313$ms ($\sim\!940$ms across three rounds); the tree gains nothing from extra cores because each leaf already has only $\sim\!7$ signatures. The tree retains a structural advantage on critical-path BLS time at any core count, but the total finality multi-core picture additionally depends on whether non-BLS stages (e.g., HotStuff's per-validator retransmit handling) parallelize, which we cannot answer without a deployed testbed.

\begin{figure*}[!htbp]
  \centering
  \includegraphics[width=\textwidth]{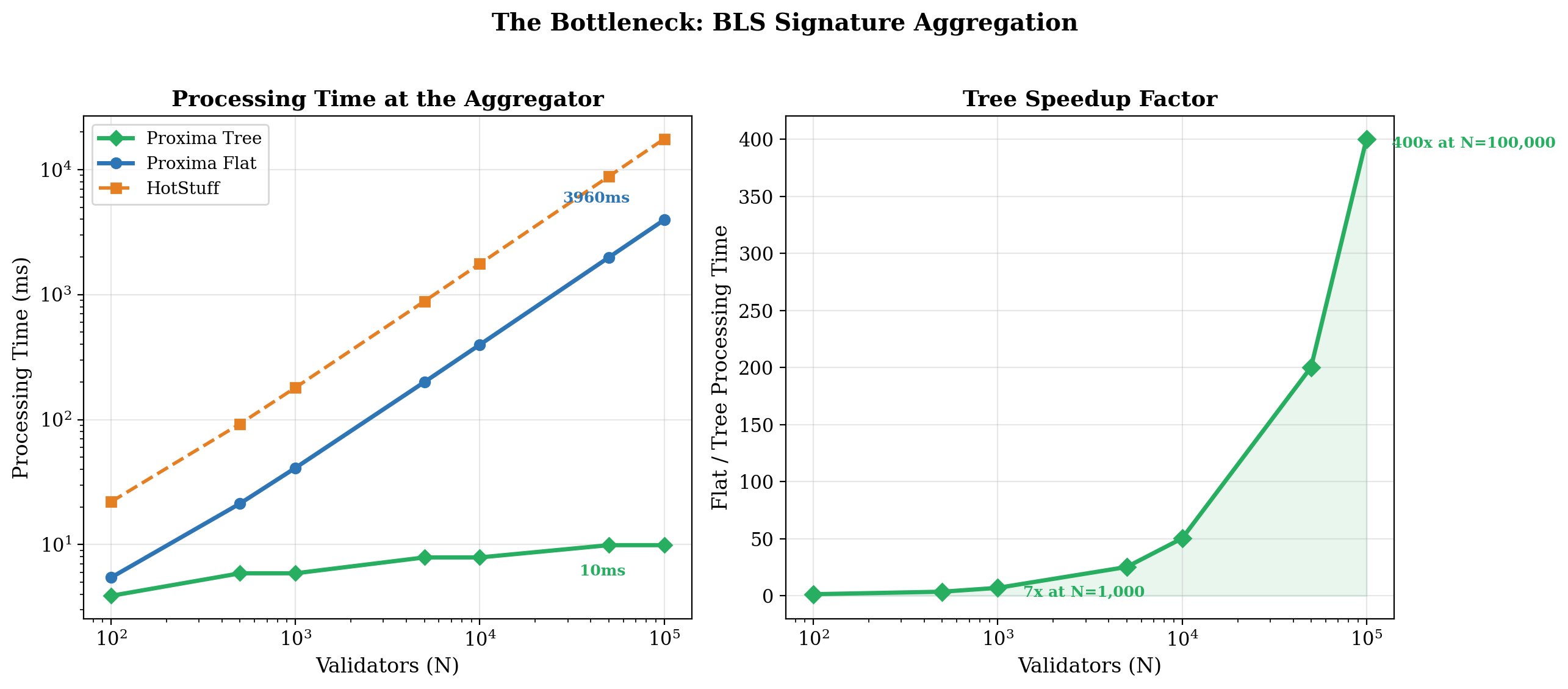}
  \caption{BLS aggregation bottleneck. Flat and HotStuff processing grows linearly with $N$; the tree stays roughly constant because each leaf processes at most branching-factor signatures in parallel. Right panel shows the tree/flat speedup factor growing with $N$.}
  \label{fig:bls}
\end{figure*}

\begin{figure*}[!htbp]
  \centering
  \includegraphics[width=\textwidth]{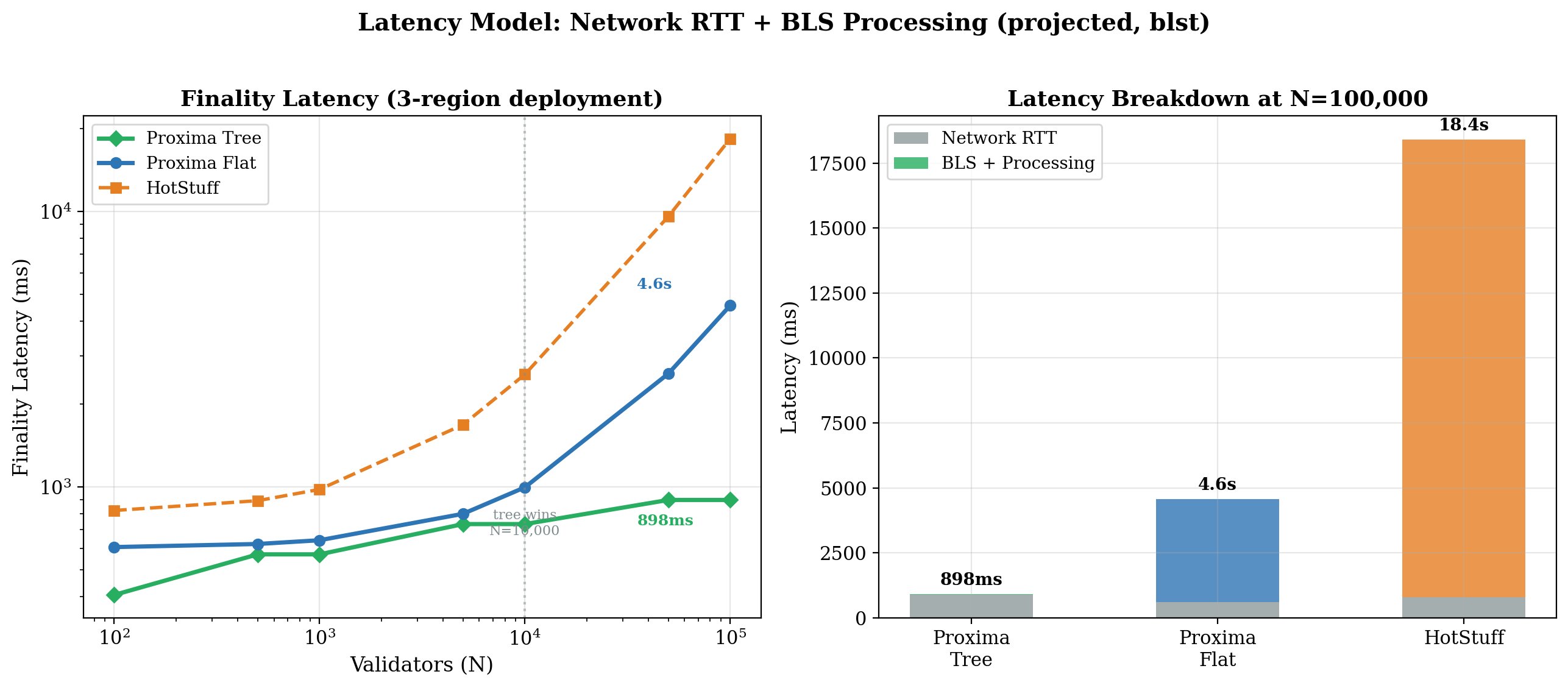}
  \caption{Total finality latency combining BLS processing and cross-region network RTT. Left: latency vs $N$. At small $N$ network RTT dominates; at large $N$ BLS processing dominates and the tree's advantage grows. Right: breakdown at $N{=}100{,}000$.}
  \label{fig:latency}
\end{figure*}

\subsection{Caveats}

\textbf{Multi-threaded baselines partially close the gap.} The tree's $9.9$ms BLS time requires leaf leaders on separate machines; sequential execution on one machine matches flat cost. Conversely, BLS aggregation in flat protocols parallelizes well: a 16-way split reduces flat aggregation to $\sim\!220$ms and each HotStuff round to $\sim\!313$ms. The tree gains little from extra cores. The tree retains an advantage on critical-path BLS time at any core count, but total finality on multi-core hardware additionally depends on whether non-BLS stages parallelize, which we cannot answer without a deployed testbed. The structural advantages, namely message complexity, smaller hierarchical groups, and fast-path finality, are core-count-independent.

\textbf{Hierarchical HotStuff variants.} Kauri~\cite{kauri} applies tree aggregation to HotStuff. A direct comparison would isolate the distance-filtering contribution from tree aggregation. This is left as future work.

\textbf{Projected, not measured.} All latency numbers use published blst constants and standard RTT estimates, not measurements from a deployed system.

\subsection{Byzantine Tolerance}

The system tolerates up to $33\%$ Byzantine, matching the standard BFT bound. In simulation: $100\%$ consensus success from $0\%$ to $33\%$, $0\%$ at $35\%$ and above. As Byzantine percentage increases, Proxima gets cheaper (fewer validators in Phase 2) while HotStuff's cost is constant.

\begin{figure*}[!htbp]
  \centering
  \includegraphics[width=0.85\textwidth]{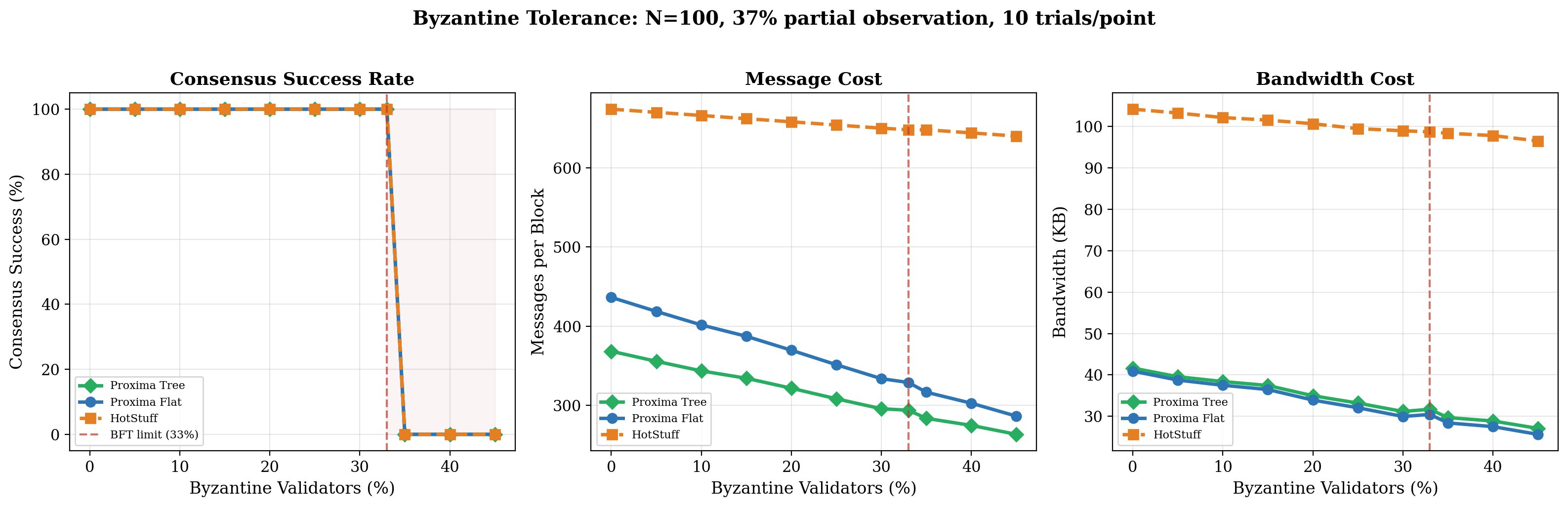}
  \caption{Byzantine tolerance and cost scaling. All three protocols achieve $100\%$ consensus success from $0$--$33\%$ Byzantine and fail at $35\%+$. As Byzantine percentage increases, Proxima's message and bandwidth costs decrease because excluded validators skip Phase 2.}
  \label{fig:byzantine-sweep}
\end{figure*}

\section{Security Analysis}

\subsection{Safety}
\label{sec:safety}

\begin{theorem}[Safety]
If fewer than $N/3$ validators are Byzantine, Proxima never finalizes two different blocks at the same height.
\end{theorem}

\begin{IEEEproof}
Finality requires a set $C$ of signed commitments with $|C| \ge 2N/3$, all containing the same block hash $h$. Assume for contradiction that blocks $B$ and $B'$ both finalize at height $k$ with commitment sets $C$ and $C'$. Since $|C| \ge 2N/3$ and $|C'| \ge 2N/3$, their intersection $|C \cap C'| \ge N/3$. Every validator in the intersection signed both $H(B)$ and $H(B')$. An honest validator signs at most one hash per height. Therefore the intersection is entirely Byzantine. But $|C \cap C'| \ge N/3$ contradicts $f < N/3$.
\end{IEEEproof}

\begin{corollary}
Phase 1 cannot compromise safety. Incorrect inclusion of a Byzantine validator means it fails to produce a valid BLS signature. Incorrect exclusion of an honest validator reduces commit count (liveness) but no incorrect block is signed.
\end{corollary}

\begin{corollary}
Tree routing cannot compromise safety. BLS aggregation is associative. A Byzantine tree node can suppress children's signatures (liveness) but cannot forge signatures.
\end{corollary}

\subsection{Liveness}

Proxima provides liveness under: (1) $f < N/3$ Byzantine; (2) the distance threshold includes at least $2N/3$ honest validators; (3) messages are eventually delivered. Condition (2) is bounded above: the probability that more than $N/3$ honest validators are excluded is at most $\exp(-0.22 N)$, negligible at any practical $N$. Under partial synchrony, a rotating aggregator with timeout ensures progress after GST.

\subsection{Collision Resistance}

A random Byzantine digest must land within Euclidean threshold $\tau$ of the reference in 8D. Treating the relevant support as a cube of side $R \approx 14$ (the half-block distance scale), the probability is the volume ratio of an 8-ball of radius $\tau$ to that cube, $\approx 4.06 \cdot (\tau / R)^8$. At $\tau = 4.9$ and $R = 14$, this is $\approx 9.1 \times 10^{-4}$, or about $0.09\%$. Empirically, Monte Carlo over $10{,}000$ trials shows $0.4\%$ of blocks have a Byzantine validator inside the threshold; the empirical figure being a few times higher than the spherical estimate is the expected sign and magnitude when the digest distribution has heavier-than-uniform tails near the threshold, and serves as a sanity check on the analytical bound. In every case Phase 2 catches it (invalid BLS commitment), so the analysis bounds liveness cost rather than safety.

\subsection{Adversarial Transaction Construction}

A natural concern is whether a Byzantine validator can craft a transaction set $T' \ne T$ whose digest lands within $\tau$ of the honest reference $D(T)$. Distance-preserving digests are deliberately not collision-resistant; a reduction to SHA-512 collision resistance would prove too much. The construction is, at our parameters, computationally easy: each $v(tx)$ is determined by SHA-512$(tx)$ and valid transactions must be signed by accounts the adversary controls (Sybils are free in a permissionless setting, so the binding constraint is gas cost per candidate). The matching condition is roughly $8 \log_2(R/\tau) \approx 12$ bits at $\tau = 4.9$, $R \approx 14$, well within laptop reach for Wagner-style $k$-list search~\cite{wagner2002gbp}. \emph{Phase 1 distance filtering therefore contributes no cryptographic security.} The construction is feasible; that is fine, because Phase 2 requires every cluster member to sign a BLS commitment to the SHA-256 block hash $H(B)$, and the Byzantine cannot forge a signature on $H(B)$ where $B$ contains $T$. Theorem~1 does not invoke any property of $D(\cdot)$. Adversarial construction therefore costs the adversary effort in exchange for no safety violation; the worst-case effect is a failed Phase 2 round, indistinguishable from any other Byzantine signature withholding.

\subsection{Bloom Filter False Positives}
\label{sec:ieee-bloom-fp}

Bloom filters admit false positives but no false negatives. A false positive causes the aggregator to skip pushing a transaction $V$ needs; $V$'s digest then exceeds $\tau$, $V$ is excluded from the cluster, and $V$ does not sign in Phase 2. The signing set shrinks by one. The chain bloom FP $\rightarrow$ exclusion $\rightarrow$ no signature is wholly a liveness chain: Theorem~1 depends only on collision-resistant hashing and BLS unforgeability, not on cluster membership, so an excluded validator cannot cause a wrong block to clear $2N/3$. The $1\%$ figure is per-lookup, not per-validator: the aggregator queries only flagged stragglers and only against $O(k)$ candidates, where $k$ is the digest-implied gap (typically $1$--$5$), giving a per-validator FP rate of $kp \approx 1$--$5\%$. Filter sizing is linear in transaction count and logarithmic in target FP rate, so scaling the filter with block size to maintain a bounded per-validator rate is a free parameter choice.

\section{Related Work}

\textbf{BFT Consensus.} PBFT~\cite{castro1999pbft} established practical BFT with $O(N^2)$ message complexity. HotStuff~\cite{yin2019hotstuff} reduced this to $O(N)$ using leader-based aggregation and BLS threshold signatures. Tendermint~\cite{buchman2016tendermint} provides a similar $O(N)$ BFT protocol. All use collision-resistant hashes for state comparison and require full state synchronization before voting. Proxima shares HotStuff's $O(N)$ voting complexity but adds Phase 1 distance-based pre-filtering and bloom-based synchronization.

\textbf{Hierarchical BFT.} Kauri~\cite{kauri} applies tree aggregation to HotStuff, reducing the leader bottleneck. Kauri's tree nodes still run per-group BFT (hash-based), requiring groups large enough for Byzantine tolerance. Proxima's distance filtering removes this requirement, enabling groups of 10 vs 128 and deeper trees (5 levels vs 2--3). Ethereum 2.0~\cite{eth2spec} uses random committees of 128 with per-committee BFT attestation; Proxima achieves similar distributed aggregation with groups $12.8\times$ smaller.

\textbf{DAG-Based Protocols.} Narwhal and Tusk~\cite{narwhal} decouple data availability from consensus via a DAG of validator proposals. Bullshark~\cite{bullshark} builds on Narwhal with a simpler consensus rule. These protocols solve a different problem but share the property that validators may have partial views. Digests could augment DAG protocols by measuring vertex agreement quality.

\textbf{Locality-Sensitive Hashing.} LSH was introduced by Indyk and Motwani~\cite{indyk1998lsh} for approximate nearest-neighbor search. SimHash~\cite{simhash} applies it to document similarity; MinHash~\cite{minhash} to Jaccard similarity; $p$-stable LSH~\cite{pstable} to $L_p$ distances. Our digest is a simple LSH scheme applied to transaction sets. The application to BFT consensus is, to our knowledge, new.

\textbf{Sharding.} Ethereum abandoned execution sharding in favor of data-availability sharding (Danksharding~\cite{danksharding}). NEAR's Nightshade~\cite{nightshade} uses receipt-based communication. Shardeum uses an address-range approach with atomic cross-shard composability. Digest-based verification offers a complementary approach: constant-cost verification per shard pair with conflict resolution only for divergent transactions.

\section{Implementation}

Proxima is implemented in Python as a working BFT blockchain with account-based state, HTTP node API, interactive demo, benchmarks, and the eight evaluation figures. BLS signatures use py-ecc~\cite{pyecc} for the live demo and hash-based mocks preserving 96-byte signature sizes for large-scale benchmarks. All message counts are tracked by incrementing a counter on each simulated send, using identical byte constants for Proxima, HotStuff, and PBFT; we cross-validated this by running both BLS settings on the small-$N$ demo and confirming identical per-block message totals. Code is available at \url{https://github.com/RyanMercier/Proxima}.

\section{Limitations and Future Work}

\textbf{Projected latency.} All latency numbers use published blst benchmarks and standard RTT estimates. A deployed testbed across 3+ cloud regions would validate the model. The Caveats subsection covers the multi-threaded picture in detail.

\textbf{No comparison against hierarchical HotStuff.} A direct comparison against Kauri would isolate distance filtering from tree aggregation. This is the most important missing evaluation.

\textbf{Adaptive adversary.} The safety proof assumes static Byzantine assignment. An adaptive adversary who corrupts validators after seeing tree assignments requires VRF-based assignment with epoch rotation, which we describe but do not evaluate.

\textbf{Cross-shard propagation assumption.} The $95\%$ propagation rate is consistent with the order-of-magnitude gossip propagation typical at multi-second block intervals but not validated on a sharded testbed.

\section{Conclusion}

The BFT literature uses collision-resistant hashing for state comparison. We identified three constraints this imposes (mandatory state synchronization, unmeasurable agreement quality, and large hierarchical committees) and showed that distance-preserving digests remove all three. At $N{=}100{,}000$, Proxima Tree uses $2.2\times$ fewer messages than HotStuff (a structural property unaffected by parallelism) and reduces cross-shard overhead by $99\%$ versus 2PC at $95\%$ propagation. Single-core latency is $\sim\!900$ms vs $\sim\!18$s for HotStuff; multi-core BLS narrows the gap considerably but the tree retains an advantage on critical-path BLS time at any core count. Safety is proved. The primitive is general and applies to any BFT protocol that compares state via hashes.


\end{document}